\documentclass[a4paper,11pt,english]{elsarticle}
\usepackage[latin1]{inputenc}
\usepackage[english]{babel}
\usepackage{verbatim}
\usepackage{psfrag}

\usepackage{calc}
\usepackage{graphicx,palatino,verbatim}

\usepackage{amsmath}

\usepackage[T1]{fontenc}

\usepackage{newlfont}

\usepackage{amsthm}

\usepackage{amssymb}

\usepackage{amsmath}

\usepackage{eufrak}

\usepackage{latexsym}

\usepackage{young}
\usepackage[mathscr]{eucal}

\usepackage{amsfonts}
\usepackage{syntonly}
\usepackage{graphicx}
\usepackage{mathptmx}
\usepackage{subfig}
\usepackage{url}
\usepackage{color}
\usepackage[letterpaper]{geometry}

\input xy
\xyoption{all}

\def\deg{\mbox{\rm deg}}

\def\Im{\mbox{\rm Im}}

\def\dim{\mbox{\rm dim}}

\def\FF{{\mathbb F}}

\def\f2{{\mathbb F}_{2}}

\newcommand{\gb}{\beta}

\newcommand{\gd}{\delta}

\newcommand{\gl}{\lambda}

\def\Im{\mathrm{Im}}

\def\INSm#1{{\color{black}#1}}

\def\INS#1{{\color{black}#1}}

\newtheorem{theorem}{Theorem}[section]

\newtheorem{remark}[theorem]{Remark}

\newtheorem{definition}[theorem]{Definition}

\newtheorem{proposition}[theorem]{Proposition}

\newtheorem{corollary}[theorem]{Corollary}

\begin{document}
  \begin{frontmatter}
  \title{A note on APN permutations in even dimension}
\author[tn]{M. Calderini}
\ead{marco.calderini@unitn.it}

\author[tn]{M. Sala}
\ead{maxsalacodes@gmail.com}

\author[tn]{I. Villa}
\ead{irene1villa@gmail.com}
\address[tn]{Department of Mathematics, University of Trento, Via Sommarive 14,
38100 Povo (Trento), Italy}
%

\begin{abstract} 
APN permutations in even dimension are vectorial Boolean functions that play a special role in the design of block ciphers.
We study their properties, providing some general results and some applications to the low-dimension cases.
In particular, we prove that none of their components can be quadratic.  \INS{For an APN vectorial Boolean function (in even dimension) with all cubic components}
we prove the existence of a component having a large number of balanced derivatives.
Using these restrictions, we obtain the first theoretical proof of the non-existence of APN permutations 
in dimension $4$. Moreover, we derive some contraints on APN permutations in dimension $6$.
\end{abstract}


\end{frontmatter}

\section{Introduction}
\label{intro}

A block cipher is a cryptographic primitive that allows the encryption/decryption of fixed-length messages
once a secret key has been shared (called the \emph{session key}). 
Given a fixed key, a block cipher can be viewed as a permutation on the message space. 
For its performance, a block cipher is designed as the composition of many efficient transformations, 
called \emph{rounds}. In any round, a \emph{round key} is derived from the session key and acts
on the message space, sometimes on the whole space (as in translation-based ciphers \cite{CGC-cry-art-carantisalaImp}, such as the AES \cite{CGC2-cry-book-daemen2013design}
or the SERPENT \cite{CGC2-cry-art-biham1998serpent}), sometimes on portions thereof (as in Feistel ciphers \cite{CGC2-cry-art-feistel1975some}).
The action of a round key is traditionally a translation, that is, any (portion of the) message is viewed as a binary vector and is summed with the round key (or XORed, in computer science language). These traditional ciphers are the
most common and include
some notable Feistel ciphers, such as  DES \INSm{\cite{fips199946}}, 
Camellia \cite{CGC2-cry-art-aoki2001camellia} and Kasumi \cite{CGC2-cry-art-unknown20103rd}, and all translation-based ciphers,
 but some alternative actions of the round keys can be found even in block ciphers which have been used
in practice, such as  IDEA \cite{CGC-cry-art-IDEA91}, SAFER \cite{CGC2-cry-art-Mass94} and GOST \cite{CGC2-cry-art-gost28147}.
Efficiency reasons explain why in traditional ciphers all-but-one round components are affine maps,
being the so-called \emph{S-box} the only exception, and the latter is then called the \emph{non-linear component}
of the cipher.

If also the \INSm{S-boxes} would be affine maps, the block cipher would be trivial to break. Unsurprisingly, the most effective attacks to traditional block ciphers start from the study of their \INSm{S-boxes} and their non-linear behavior, as explained below.
In particular, the so-called \emph{differential cryptanalysis} \cite{CGC2-cry-art-biham1991differential} has proved especially effective.
The most basic version of this attack successfully applies when two plaintexts with a (known) fixed difference 
lead after the last-but-one round to outputs whose difference takes a (known) value with a probability significantly
higher than the average.
To minimize the success probability of this attack, the theory of vectorial Boolean functions (\cite{CGC-cd-art-carlet10})
has identified an ideal property for the involved \INSm{S-box}, that is, to be an Almost Perfect Nonlinear function
 (or {APN} for short).
Relevant definitions and properties of APN functions can be found in Section 2. It is important to note
that the APN functions used in translation-based ciphers have to be bijective, that is, they must
be permutations. 
\INSm{In the case of the ciphers AES and SERPENT the S-boxes used are not APN. As regards AES, there is no known example of an 8-bit APN permutation, while for SERPENT there does not exist a 4-bit APN permutation.}\\
 However, in Feistel ciphers is not necessary \INSm{that the S-boxes are invertible} (not even in DES), although they are in some, 
like in Kasumi. \INSm{In Kasumi the S-boxes are APN permutation, as they are defined over an odd dimensional space.}

Although much is known for APN permutations in odd dimension, implementative reasons make a case 
for the use of APN functions in even dimension, especially when the dimension is a power of $2$.
Unfortunately, little is known at present for these cases and what is known relies heavily
on computer checking:
\begin{itemize}
\item It is known that there is no APN permutation in dimension 4 (the first non-trivial case), but the proof
relies on extensive computations providing no theoretical insight on the reasons behind their non-existence.
\item It is known that there is at least one APN permutation with dimension 6, called the Dillon's permutation
(\cite{CGC2-cry-art-browning2010apn}). Interestingly, prior to \cite{CGC2-cry-art-browning2010apn} it was conjectured
that no APN permutations in even dimension could exist, since nobody could find any example \cite{CGC-cry-art-Hou2}.
\item
It is known that any cubic APN permutation in dimension 6 must be CCZ-equivalent to the Dillon's permutation, 
which is itself CCZ-equivalent to a quadratic function, but again this proof is of a purely computational nature
\cite{CGC2-cry-art-lan}.
\item All computationally-found APN permutations in dimension $6$ lack a quadratic component, again
         with no hint to as why \cite{CGC2-cry-art-lan}.
\item Very little is known on a putative APN permutation in dimension $8$ or higher.
\end{itemize}

In this paper we present some advances in the theoretical understanding of properties for APN permutations
of even dimension. After having provided some notation and preliminaries in Section \ref{sec:1}, we claim
the following main results:
\begin{itemize}
\item in Section \ref{sec:2}, any APN permutation must lack a partially-bent component; this implies
         that it must also lack a quadratic component;
\item in Section \ref{sec:3}, any cubic APN function (not necessarily a permutation) must have a component
         with a large number of balanced derivatives; moreover, we classify cubic Boolean functions
        in dimension $4$ according to their number of balanced derivatives;
\item in Section \ref{sec:4}, we derive two immediate consequences of our previous results; the first relates
         to dimension $4$ and is the first-ever theoretical proof of the non-existence of an APN permutation;
         the second relates to dimension $6$ and is the theoretical explanation of the component degrees for
        all known APN permutations.
\end{itemize}

\section{Preliminaries}\label{sec:1}
\INSm{We will denote by $\FF_2$ the finite field with two elements.} Let $m\ge 1$, in the sequel we consider Boolean function $f$ from \INSm{$(\FF_2)^m$} to \INSm{$\FF_2$} and only vectorial Boolean function $F$ from \INSm{$(\FF_2)^m$} to \INSm{$(\FF_2)^m$}. Without loss of generality we will assume $f(0) = 0$ ($F(0)=0$).\\
We denote the \emph{derivative} of $f$ in the direction of \INSm{$a\in (\FF_2)^m$} by $D_a{f}(x)=f(x+a)+f(x)$ and the \emph{image} of $f$ by $\Im(f) = \{f(x) \mid \,x \in \INSm{(\FF_2)^m}\}$ (similarly for \INS{vectorial Boolean functions}).\\
 Let $\gl\in\INSm{(\FF_2)^m}$, we denote by $F_\gl$ the \emph{component} $\sum_{i=1}^m \gl_if_i$ of $F$, where $f_1,\dots,f_m$ are the coordinate functions of $F$. Note that for a \INS{vectorial Boolean function} $F$ we have $D_aF_\gl=(D_aF)_\gl$.

Let $f$ be a \INS{Boolean function}, if $\deg(f)=0,1,2,3$, then $f$ is, respectively, \emph{constant}, \emph{linear}, \emph{quadratic}, \emph{cubic}.
Let $F$ be a \INS{vectorial Boolean function}, we say that $F$ is \emph{quadratic} if $\max_{\lambda \not= 0} \{\deg(F_\lambda)\} =2$,
\emph{cubic} if $\max_{\lambda \not= 0} \{\deg(F_\lambda)\} =3$. 
If all non-zero components of $F$ are quadratic then we say that $F$ is a \emph{pure quadratic}.
If all non-zero components of $F$ are cubic then we say that $F$ is a \emph{pure cubic}.
\begin{definition}
Let $F$ be a \INS{vectorial Boolean function}, for any $a,b\in\INSm{(\FF_2)^m}$  we define
$$
\gd_F(a,b)=|\{x\in\INSm{(\FF_2)^m} \mid\,D_aF(x)=b\}|.
$$
The \emph{differential uniformity} of $F$ is 
$$
\INSm{\gd(F)}=\max_{a,b \in\INSm{(\FF_2)^m}\\
a\neq 0}\gd_F(a,b)\,.
$$
Those functions  with $\INSm{\gd\INSm{(F)}}=2$ are said \emph{Almost Perfect Nonlinear (APN)}.
\end{definition}

\noindent We denote by  $\mathcal{F}(f)$ the following value related to
the Fourier transform of a \INS{Boolean function} $f$:
$$
\mathcal{F}(f)=\sum_{x\in \INSm{(\FF_2)^m}}(-1)^{f(x)}=2^m-2\INSm{\mathrm{w}_H}(f),
$$
where $\INSm{\mathrm{w}_H}(f)$ is the Hamming weight of $f$, i.e. the number of $x$ 
such that $f(x)=1$. The function is said to be \emph{balanced}
if and only if $\mathcal{F}(f)=0$.
\\

A necessary condition for a \INS{vectorial Boolean function} $F$ to be APN
was provided by Nyberg in \cite{CGC2-cry-art-nyberg1995s}. This condition involves the
derivatives of the components of $F$. It was proved by Berger \emph{et al.} \cite{CGC2-cry-art-berger2006almost} that this condition is also sufficient.

\begin{proposition}[\cite{CGC2-cry-art-nyberg1995s,CGC2-cry-art-berger2006almost}]\label{th:F-APN}
Let $F$ be a \INS{vectorial Boolean function}. Then, for any non-zero $a\in\INSm{(\FF_2)^m}$
$$
\sum_{\gl\in\INSm{(\FF_2)^m}}\mathcal{F}^2(D_aF_\gl)\ge 2^{2m+1}.
$$
Moreover, $F$ is APN if and only if for all non-zero $a\in\INSm{(\FF_2)^m}$
$$
\sum_{\gl\in\INSm{(\FF_2)^m}}\mathcal{F}^2(D_aF_\gl)= 2^{2m+1}.
$$
\end{proposition}

We recall the following non-linearity measures for \INS{vectorial Boolean functions}, as introduced in \cite{CGC-cry-art-carantisalaImp,CGC2-cry-art-fontanari2012weakly}:

$$
n_i(F)=|\{\gl\,\in\INSm{(\FF_2)}^m \setminus\{0\}\mid\,\deg(F_\gl)=i\}|,
$$

$$
\hat n (F)=\max_{a\in\INSm{(\FF_2)}^m\setminus\{0\}}|\{\gl\,\in\INSm{(\FF_2)^m} \setminus\{0\}\mid\,\deg(D_aF_\gl)=0\}|,
$$

$$
\bar{\gd}(F)=\max_{a\in\INSm{(\FF_2)}^m\setminus\{0\}}\left(\mathrm{min} \left\{\gd\in\mathbb{N}\mid |\Im(D_aF)|>\frac{2^{m-1}}{\gd}\right\}\right).
$$
\INSm{$\bar{\gd}(F)$ is the \emph{weakly differential uniformity} of $F$}. If $\bar{\gd}(F)=2$ then $F$ is said \emph{weakly-APN}.\\
For a \INS{vectorial Boolean function} we report the following result.

\begin{theorem}[\cite{CGC-cry-art-carantisalaImp,CGC2-cry-art-fontanari2012weakly}]\label{th:nhat}
Let $F$ be a \INS{vectorial Boolean function}, then
\begin{itemize}
\item[1)] $\gd(F)\ge \bar{\gd}(F)$,\\
 in particular if $F$ is APN then it is weakly-APN.
\item[2)] If $F$ is weakly-APN, then $\hat{n}(F)\le 1$,\\
in particular $F$ APN implies $\hat{n}(F)\le 1$.
\end{itemize}
\end{theorem}
The following theorem is well-known.
\begin{theorem}[\cite{CGC-cry-art-carlet2010vectorial}]\label{th:linapn}
Let $F$ be APN, then $n_1(F)=0$.
\end{theorem}

Theorem \ref{th:linapn} cannot be extended to weakly-APN functions, since from the classification of bijective \INS{vectorial Boolean functions} for $m=$4 there is one affine equivalent class of weakly-APN functions with $n_1=1$.
\\

Finally we recall some results on quadratic and partially bent \INS{Boolean functions}. \\
The following two results can be found in \cite{CGC-cd-book-macwilliamsI} Chapter 15.

\begin{proposition}[\INSm{\cite{CGC-cd-book-macwilliamsI}}]\label{prop:quad}
Every quadratic function is affinely equivalent to:
\begin{itemize}
\item $x_1x_2 +\dots+x_{2l-1}x_{2l} +x_{2l+1}$ (where $l \le \frac{m-1}{2}$) if it is balanced, 
\item $x_1x_2 +\dots+x_{2l-1}x_{2l}$ (where $l \le m/2$) if it has weight smaller than $2^{m-1}$, 
\item $x_1x_2 +\dots+x_{2l-1}x_{2l} + 1$ (where $l \le m/2$) if it has weight greater than $2^{m-1}$.
\end{itemize}
\end{proposition}

Denote by $V(f)=\{a\mid D_af \text{ is constant}\}$ the set of \emph{linear structures} of a \INS{Boolean function} $f$. Observe that $V(f)$ is a vector subspace.
\begin{proposition}[\INSm{\cite{CGC-cd-book-macwilliamsI}}]
Any quadratic \INS{Boolean function} $f$ is balanced if and only if its
restriction to $V(f)$ is not constant. If it is not balanced, then its weight equals $2^{m-1} \pm 2^{\frac{m+k}{2} -1}$
where $k$ is the dimension of $V(f)$.
\end{proposition}

\begin{remark}\label{rm:quadratic}
The proposition above implies that for any non-balanced quadratic \INS{Boolean function} we have $\mathcal{F}(f) =\pm 2^\frac{m+k}{2}$.
\end{remark}

\INS{We report now the definition of partially bent function, which was introduced in \cite{CGC2-cry-art-carlet1993partially}.}

\begin{definition}
A \INS{Boolean function} $f$ is \emph{partially bent} if there exists a linear subspace $\bar{V}(f)$ of \INSm{$(\FF_2)^m$} such that the restriction of $f$ to $\bar{V}(f)$ is affine and the restriction of $f$ to any complementary subspace $U$ of $\bar{V}(f)$, $\bar{V}(f)\oplus U =\INSm{(\FF_2)^m}$, is bent.
\end{definition}

\begin{remark}\label{rk:lin}
If $f$ is partially bent, then $f$ can be represented as a direct sum of the restricted functions, i.e., $f(y + z) = f(y) + f(z)$, for all $z\in \bar{V}(f)$ and $y\in U$. Moreover we can deduce from Proposition \ref{prop:quad} that any quadratic function is partially bent.
\end{remark}

\begin{remark}\label{rk:U}
If $f$ is partially bent, the space $\bar{V}(f)$ is formed by the linear structures of $f$, which is $\bar{V}(f)=V(f)$. \INSm{Indeed, let $a\in\bar{V}(f)\setminus\{0\}$ and $x\in (\FF_2)^m$. Then $x=y+z$ for some $z\in \bar{V}(f)$ and $y\in U$. So, from Remark \ref{rk:lin}, we have
$$
D_af(x)=f(x+a)+f(x)=f(y+z+a)+f(y+z)=f(y)+f(z)+f(a)+f(y)+f(z)=f(a),
$$
that implies $\bar{V}(f)\subseteq V(f)$.\\
 Now, suppose that there exists $a\in{V}(f)\setminus\bar{V}(f)$. Without loss of generality $a\in U$. By definition $f_{|_U}$ (the restriction of $f$ to $U$) is bent. This implies $D_af_{|_U}$ is balanced, but this is not possible as $D_af(x)=f(a)$ for all $x\in (\FF_2)^m$. Then $\bar{V}(f)= V(f)$.\\
Moreover, since bent functions exist only in even dimension, $\mathrm{dim}(U)=m-\mathrm{dim}(V(f))$ is even. Which implies that if $m$ is even, the dimension of $V(f)$ is even. In particular $V(f)=\{0\}$ if and only if $f$ is bent. This implies that if $f$ is balanced and $m$ is even then $\mathrm{dim}( V(f))\ge 2$.} 
\end{remark}

\section{Properties of APN permutations}\label{sec:2}

Our first result holds for any dimension and comes directly from the facts in previous section.
\begin{theorem}
Let $m\geq 3$ and let $F$ be an APN permutation of dimension $m$.
There are only two cases:
\begin{itemize}
\item $\hat{n}(F)=0$, which implies that for any $\gl\neq 0$ $F_\lambda$ is not partially bent and so $n_1(F)=n_2(F)=0$;
\item $\hat{n}(F)=1$, for which it is possible that there is a $\gl\not=0$ such that $F_\gl$ is partially bent, and so
                    $n_1(F)=0$ and $n_2(F)\geq 0$.
\end{itemize}
\end{theorem}
\begin{proof}
From Theorem \ref{th:nhat} and Theorem \ref{th:linapn} we have that $\hat{n}(F)\le 1$ and $n_1(F)=0$. If $\hat{n}(F)=0$ then any non-zero component $F_\lambda$ has $V(F_\lambda)=\{0\}$. This implies that if $F_\lambda$ is partially bent, then $F_\lambda$ is bent, see Remark \ref{rk:U}, but it is not possible as $F$ is a permutation. 

If $\hat{n}(F)=1$ then there exists $F_\lambda$ with $V(F_\lambda)\ne\{0\}$, and it could be partially bent (in particular it could be quadratic). 
\end{proof}

The condition on bijection for $F$ is essential, otherwise the case $\hat n=0$ may have bent components. As there are examples of quadratic APN permutations for any odd dimension, the previous theorem cannot be improved. 

The case of an APN permutation $F$ with $m$ even is quite different and we will restrict to it from now on.\\
As shown in \cite{CGC2-cry-art-seberry1994pitfalls} there is no APN quadratic permutation over $\INSm{(\FF_2)^m}$ for $m$ even, that is $n_2(F)\le 2^{m-1}-1$. This result was extended by Nyberg \cite{CGC2-cry-art-nyberg1995s} showing that an APN permutation $F$ cannot have all components partially bent (for $m$ even). \\
Moreover, \INSm{ in \cite{GCG-cry-art-onweak} the authors give some properties on the the components of a weakly-APN permutation in even dimension. In particular they study partially-bent and quadratic components, obtaining that the number of the quadratic components of a weakly-APN permutation can be at most $2^{m-2}-1$ (\cite{GCG-cry-art-onweak} Proposition 4). As an APN function is weakly-APN, we have $n_2(F)\leq 2^{m-2}-1$.\\}
In the remainder of this section we will prove that \emph{no component} of $F$ is partially bent (quadratic).
\\

We start with the following proposition.

\begin{proposition}\label{prop:cost}
Let $F$ be an APN permutation over $\INSm{(\FF_2)^m}$, with $m$ even. If there are $a,\gl\in \INSm{(\FF_2)^m}\setminus\{0\}$ such that $D_aF_\gl$ is constant, then $D_aF_\gl=1$.
\end{proposition}
\begin{proof}
Suppose that there exist non-zero $a,\gl\in \INSm{(\FF_2)^m}$ such that $D_aF_\gl=0$. Without loss of generality we can suppose that $F_\gl=f_1$, thus we have
$$
\Im(D_aF)=\{(0,y_2,\dots,y_{m})\mid y_i\in \FF_2\}.
$$
Being $F$ APN we have $|\Im(D_aF)|=2^{m-1}$, so $0$ has to lie in $\Im(D_aF)$, contradicting the fact that $F$ is a permutation.
\end{proof}

\begin{theorem}\label{th:part}
Let $F$ be an APN permutation over $\INSm{(\FF_2)^m}$, with $m$ even, then no non-zero component of $F$ is partially bent.
\end{theorem}
\begin{proof}
Suppose that $F_\gl$ is partially bent, for some $\gl \in \INSm{(\FF_2)^m}\setminus\{0\}$. From Remark \ref{rk:U} and being $F$ a permutation, we have that the space of the linear structure of $F_\gl$ has at least dimension $2$. Let $a_1$ and $a_2$ be two distinct non-zero vectors of $V(F_\gl)$. Let $x\in\INSm{(\FF_2)^m}$, from Proposition \ref{prop:cost} and Remark \ref{rk:U} we have $D_{a_1}F_\gl(x)=F_\gl(a_1)=1$ and $D_{a_2}F_\gl(x)=F_\gl(a_2)=1$. This implies $a_1+a_2\in V(F_\gl)$, $a_1+a_2\neq 0$ and
$$
D_{a_1+a_2}F_\gl(x)=F_\gl(a_1+a_2)=F_\gl(a_1)+F_\gl(a_2)=0\quad \mbox{(for all }x\in\INSm{(\FF_2)^m}\mbox{)}.
$$
But $D_{a_1+a_2}F_\gl=0$ contradicts Proposition \ref{prop:cost}.
\end{proof}

\begin{corollary}\label{cor:quad}
Let $F$ be an APN permutation over $\INSm{(\FF_2)^m}$, for $m$ even. Then $n_2(F)=0$.
\end{corollary}

\begin{remark}
As observed by C. Carlet and L. Budaghyan in a private communication, Theorem \ref{th:part} cannot be extended to the non-existence of plateaued components, since Dillon's APN permutation does have some.
\end{remark}

\INSm{\begin{remark}\label{rk:bent_lin}
Let $f$ be a bent Boolean function and $\ell$ be a linear Boolean function. Then $f+\ell$ is bent. Indeed, let $a\in (\FF_2)^m\setminus\{0\}$. Thus we have
$$
D_a(f+\ell)(x)=D_af(x)+D_a\ell(x)=D_af(x)+c,
$$
where $c\in\FF_2$. As $D_af$ is balanced, then $D_a(f+\ell)$ is balanced. This implies that $f+\ell$ is bent.\\
Moreover, we immediately have that if $f$ is partially bent and $\ell$ is a linear Boolean function, then $f+\ell$ is partially bent.
\end{remark}
\INSm{We recall that two functions $F$ and $F'$ are called EA-equivalent if there are an affine mapping $L$ and function $G$ affinely equivalent to $F$, such that $F'=G+L$}.\\
From Remark \ref{rk:bent_lin} we obtain the following.
\begin{proposition}\label{prop:quad1}
Let $F$ be a vectorial Boolean function with partially bent components. If $F'$ is EA-equivalent to $F$, then $F'$ has partially bent components.
\end{proposition}
\begin{proof}
$F'=G+L$ for some $G$ affine equivalent to $F$ and $L$ affine map of $(\FF_2)^m$. Thus the component $\gl$ of $F'$ is $F'_\gl=G_\gl+L_\gl$ for all $\gl\in (\FF_2)^m$. As $F$ has a  partially bent component, then also $G$ has a partially bent component (it is affine equivalent to $F$). Let $G_\gl$ be partially bent, then from Remark \ref{rk:bent_lin} we have that $G_\gl+L_\gl$ is partially bent.
\end{proof}}

\INSm{For $m$ even and $\gcd(m,i)=1$ the following two families of APN functions were constructed in \cite{lilya1}, \cite{lilya2}
$$x^{2^i+1}+(x^{2^i}+x+1)\mathrm{Tr}(x^{2^i+1}),$$
 $$x^3+\mathrm{Tr}(x^9)+(x^2+x+1)\mathrm{Tr}(x^3),$$
where  $\mathrm{Tr}(x)$ denotes the trace function from $\FF_{2^m}$ into $\FF_2$.
It was proven in \cite{li2013nonexistence} that the first one is not EA-equivalent to permutations and, at the best of our knowledge, a similar result is missing for the second one.
However, both functions have quadratic components which implies, according to Corollary \ref{cor:quad} and Proposition \ref{prop:quad1}, that both of them are not EA-equivalent to permutations.}\\
\INSm{More generally, since EA-equivalence preserves partially bent components, then the following holds:
\begin{corollary}\label{cor:ea}
Let $m$ be even and $F$ be an APN function over $(\FF_2)^m$ having a partially bent (quadratic) component. Then $F$ is EA-inequivalent to any permutation.
\end{corollary}}

%
%
\section{On cubics in even dimension}\label{sec:3}

In this section we are interested in cubic \INS{Boolean functions} and cubic \INS{vectorial Boolean functions}.\\

Given a \INS{Boolean function} $f$, we are interested in counting the number of its derivatives that are balanced, that is the cardinality of 
$\Gamma(f)=\{a\in \INSm{(\FF_2)^m}\mid D_af \text{ is balanced}\}$. Observe that if $f$ is quadratic, then its derivatives
can be either linear functions (which are balanced) or constant functions, which are non-balanced.
For a cubic function, the situation is less obvious.
The following theorem presents an estimate of $\Gamma$ for a component of a pure-cubic APN function.
\INS{\begin{theorem}\label{prop:gamma}
Let $F$ be a cubic APN \INS{vectorial Boolean function} over $\INSm{(\FF_2)^m}$, with $m$ even.
Then there exists $\lambda\neq{0}$ such that $$|\Gamma(F_\lambda)|\geq2^m-2^{m-2}-1.$$
\end{theorem}}
\begin{proof}
Consider any $a\in\mathbb{F}^m$, $a\neq 0$. For any component $F_\gl$ we are interested in the following integers: $\deg (D_aF_\gl)$, where clearly $0\leq\deg (D_aF_\gl)\leq 2$, $\mathcal{F}(D_aF_\gl)$, and $k_\lambda=\dim(V(D_aF_\gl))$. We note that $k_\lambda\geq 2$, because $D_a(D_aF_\gl)=0$ and so $D_aF_\gl$ cannot be bent (see Remark \ref{rk:U}). We can have the following cases:

\begin{itemize}
\item[1)] $\deg (D_aF_\gl)=0$, then we have $D_aF_\gl =0,1$ and so $\mathcal{F}(D_aF_\gl)=\pm2^m=\pm2^{\frac{m+m}{2}}$.
\item[2)] $\deg (D_aF_\gl)=1$, then $D_aF_\gl$ is balanced and so $\mathcal{F}(D_aF_\gl)=0$.
\item[3)] $\deg (D_aF_\gl)=2$ and $D_aF_\gl$ is balanced, so $\mathcal{F}(D_aF_\gl)=0$.
\item[4)] $\deg (D_aF_\gl)=2$ and $D_aF_\gl$ is not balanced, so from Remark  \ref{rm:quadratic} we have $\mathcal{F}(D_aF_\gl) =\pm 2^\frac{m+k_\gl}{2}$.
\end{itemize}
Let $s_\gl = \frac{m+k_\gl}{2}$, noting that $k_\gl\ge 2$, then $\frac m 2 +1 \le s_\gl\le m$. Observe that in case 1 and 4 we have $\mathcal{F}^2(D_aF_\gl)=2^{2s_\lambda}$, while in
case 2 and 3 we have $\mathcal{F}^2(D_aF_\gl)=0$. Since $a\neq0$ and $\mathcal{F}^2(D_aF_0)=2^{2m}$, from  Proposition \ref{th:F-APN}, we have
 $$\sum_{\lambda\neq0}\mathcal{F}^2(D_aF_\gl)=2^{2m}.$$
 
 We will now consider functions $D_aF_\lambda$, with $\lambda$ and $a$ varying freely in $\INSm{(\FF_2)^m}\setminus\{0\}$.
 
Now let $\Delta_a = \lbrace \lambda\neq{0}\mid\mathcal{F}(D_aF_\gl)\neq0 \rbrace$ and $\overline{\Delta_a} = \lbrace \lambda\neq{0}\mid\mathcal{F}(D_aF_\gl)=0 \rbrace$, i.e. $\Delta_a\cup\overline{\Delta_a}=\INSm{(\FF_2)^m}\setminus\{0\}$.
Hence in our case we have 
$$2^{2m}=\sum_{\lambda\neq {0}}\mathcal{F}^2(D_aF_\gl)=\sum_{\lambda\in\Delta_a}\mathcal{F}^2(D_aF_\gl)=\sum_{\lambda\in\Delta_a}2^{2s_\lambda}.$$

Since $s_\gl\ge \frac m 2 +1$, we have
 $$2^{2m}\geq \sum_{\lambda\in\Delta_a}2^{2(\frac{m}{2}+1)}=2^{m+2}|\Delta_a|,\quad\text{ and so }
|\Delta_a|\leq 2^{m-2}.$$
We have thus proved that 
 $$|\Delta_a|\leq2^{m-2} \mathrm{\ and\ } |\overline{\Delta_a}|\geq 2^m-2^{m-2}-1.$$
Note that $\Gamma(F_\lambda) = \lbrace a \mid \mathcal{F}(D_aF_\gl)=0\rbrace$.
Assuming now that for all $\lambda\neq{0}$ we have $|\Gamma(F_\lambda)|<2^m-2^{m-2}-1$, we would have
$$(2^m-1)\cdot(2^m-2^{m-2}-1)\leq\sum_{a\neq {0}}|\overline{\Delta_a}|=\sum_{\lambda\neq {0}}|\Gamma(F_\lambda|)<(2^m-1)\cdot(2^m-2^{m-2}-1),$$
and this is impossible. Thus there exists $\lambda\neq{0}$ such that $|\Gamma(F_\lambda)|\geq2^m-2^{m-2}-1$.
\end{proof}

Observe that the previous theorem holds, when $m$ is even,  for any cubic APN function $F$ such that 
$F:\INSm{(\FF_2)^m} \to \INSm{(\FF_2)^m}$, even if $F$ is not a permutation (and even if $F$ is a pure-cubic). 
\\

It turns out that $\Gamma(f)$ suffers a strong constraint when we specialize to the case $m=4$,
as next theorem shows. 

\begin{theorem}\label{lm:gamma}
Let $f:(\FF_2)^4\rightarrow \FF_2$ be a cubic \INS{Boolean function}. Then $$|\Gamma(f)|=|\lbrace a \mid \mathcal{F}(D_af)=0\rbrace|< 11.$$
\end{theorem}
\begin{proof}
Suppose that $|\Gamma(f)|\ge 11$. Since $|\Gamma(f)|>8$, $\Gamma(f) $ contains 4 linearly independent vectors. Without loss of generality we can assume $e_1, e_2, e_3, e_4$ (the standard basis) belong to $\Gamma(f)$. We will implicitly use in the remainder of this proof that $D_{e_1}(f)$,..., $D_{e_4}(f)$ are balanced for our $f$. 

The \INS{Boolean function} $f$ can be written as $f=\mathrm{supp}_0(f)+\mathrm{supp}_1(f)+\mathrm{supp}_2(f)+\mathrm{supp}_3(f)$, where $\mathrm{supp}_i(f)$ contains only terms of
degree $i$ for $i=1,2,3$.  Note that $\mathrm{supp}_0(f)+\mathrm{supp}_1(f)$, that is, the linear part does not influence the balancedness of any derivative of $f$, and so we can consider $f=\mathrm{supp}_2(f)+\mathrm{supp}_3(f)$.
We will find a contradiction depending on the following cases, which are characterized by the weight of $\mathrm{supp}_3(f)$, i.e. $|\mathrm{supp}_3(f)|$, which can only vary in $1 \leq |\mathrm{supp}_3(f)|\leq 4$ since $m=4$ and $\deg(f)=3$.\\
In the following cases we will write \INS{Boolean functions} as polynomials that may contain squares, using the standard notation
of viewing them implicitly in the quotient ring $\FF_2[x_1,x_2,x_3,x_4]/\langle x_1^2-x_1,\dots,x_4^2-x_4\rangle$.
Also, since we will apply affine transformations we will sometimes obtain a \INS{Boolean function} $g$ such that $g(0)=1$.
However, we will write that $g$ is \emph{equivalent} to $g'$ if $g$ is affine equivalent to either $g'$ or
$g'+1$. This equivalence notion preserves balancedness and it is appropriate for our goals.

\begin{itemize}
\item[\bf i)] {\bf $|\mathrm{supp}_3(f)|=4$}

$$
f(x)=x_1x_2x_3+x_1x_2x_4+x_1x_3x_4+x_2x_3x_4+\sum_{i<j}\gb_{ij}x_ix_j.
$$
We have already assumed that the standard basis belongs to $\Gamma(f)$, so $\Gamma(f)$ contains all vectors of weight 1. There are other 7 vectors in $\Gamma(f)$ and in $(\FF_2)^4$ there are   five vectors of weight 3 or 4. Thus in $\Gamma(f)$ there is at least a vector of weight 2. A permutation  of the coordinates will not change our situation, that is, $\Gamma$ will still contain the vectors with the prescribed weights, and so without loss of generality we can assume that $(1100)\in\Gamma$.
If we derive $f$ in the direction of $e_1$ we obtain
\begin{align*}
D_{e_1}f(x)  = & x_2x_3+x_2x_4+x_3x_4+\sum_{j\ne1}\beta_{1j}x_j\\
 = & (x_2+x_3+\beta_{14})(x_3+x_4+\beta_{12})+(1+\beta_{12}+\beta_{13}+\beta_{14})x_3+\beta_{12}\beta_{14},
\end{align*}
that is equivalent to $x_1  x_2+(1+\beta_{12}+\beta_{13}+\beta_{14})x_3$. From Proposition \ref{prop:quad} we have that $D_{e_1}f$ is balanced if and only if $\beta_{12}+\beta_{13}+\beta_{14}=0$.

Similarly from the derivative $D_{e_2}f$ we get that $\beta_{12}+\beta_{23}+\beta_{24}=0$.\\
Now, if we derive the function in the direction of $a=(1100)$ we obtain
\begin{align*}
D_af(x) = &  x_1x_3+x_2x_3+x_3+x_1x_4+x_2x_4+x_4+x_3x_4+x_3x_4\\
& \ +\beta_{12}(x_1+x_2+1)+\beta_{13}x_3+\beta_{14}x_4+\beta_{23}x_3+\beta_{24}x_4\\
 = &  (x_1+x_2+1+\beta_{13}+\beta_{23})(x_3+x_4+\beta_{12})\\
 & \  +x_4(\beta_{13}+\beta_{23}+\beta_{14}+\beta_{24})+\beta_{12}(\beta_{13}+\beta_{23}),
\end{align*}
which is equivalent to $x_1  x_2+(\beta_{13}+\beta_{23}+\beta_{14}+\beta_{24})x_3$.
 The obtained \INS{Boolean function} is balanced if and only if $\beta_{13}+\beta_{23}+\beta_{14}+\beta_{24}=1$, which contradicts the previous conditions, since
 $$
 0+0=\beta_{12}+\beta_{13}+\beta_{14}+\beta_{12}+\beta_{23}+\beta_{24}=\beta_{13}+\beta_{23}+\beta_{14}+\beta_{24}.
 $$ 
\item[\bf ii)] {\bf$|\mathrm{supp}_3(f)|=3$}

Without loss of generality we can consider 
$$
f(x)=x_1x_2x_3+x_1x_2x_4+x_1x_3x_4+\sum_{i<j}\gb_{ij}x_ix_j.
$$

Considering the derivative in the direction of $e_2$ we have
\begin{align*}
D_{e_2}f(x) = & x_1x_3+x_1x_4+\sum_{j\ne2}\beta_{2j}x_j\\
= & (x_1+\beta_{23})(x_3+x_4+\beta_{12})+(\beta_{23}+\beta_{24})x_4+\beta_{12}\beta_{23}.
\end{align*}
As before, from Proposition \ref{prop:quad} we have $\beta_{23}+\beta_{24}=1$.
Similarly, we have $\beta_{23}+\beta_{34}=1$ from $D_{e_3}f$, and $\beta_{24}+\beta_{34}=1$ from $D_{e_4}f$. The three conditions on $\gb_{ij}$'s cannot simultaneously hold. 

\item[\bf iii)] {\bf $|\mathrm{supp}_3(f)|=2$}

In this case we can assume without loss of generality that $f$ is given by
$$f(x) = x_1x_2x_3+x_1x_2x_4+\sum_{i<j}\beta_{ij}x_ix_j.$$
We will consider first its derivatives in the direction of the $e_1,e_2,e_3$ and, as usual, we will determine the conditions that preserve the linear part (Proposition \ref{prop:quad} first case).
\begin{align*}
D_{e_1}f(x) = & x_2x_3+x_2x_4+\sum_{j\ne1}\beta_{1j}x_j\\
 = & (x_2+\beta_{13})(x_3+x_4+\beta_{12})+(\beta_{13}+\beta_{14})x_4+\beta_{12}\beta_{13},
\end{align*}
 then $\beta_{13}+\beta_{14}=1$.\\
Similarly we have $\beta_{23}+\beta_{24}=1$ from $D_{e_2}f$.\\
Deriving $f$ in the direction of $e_3$ we have
\begin{align*}
D_{e_3}f(x) = & x_1x_2+\sum_j\beta_{3j}x_j\\
= & (x_1+\beta_{23})(x_2+\beta_{13})+\beta_{34}x_4+\beta_{13}\beta_{23}.
\end{align*}
 Then $\beta_{34}=1$.\\
 
 The collected conditions are 
 $$
 \gb_{14}=\gb_{13}+1,\,\gb_{24}=\gb_{23}+1,\,\gb_{34}=1.
 $$
Let  $a\in\mathbb{F}^m$ , $a = (\gamma,\gamma+1,{\xi},{\xi}+1)$ with $\gamma,{\xi}\in \FF_2$.
Consider the derivative in the direction of $a$.
\begin{align*}
D_{a}f(x) = &   
({\xi} x_1x_2+(\gamma+1)x_1x_3+(\gamma+1){\xi} x_1+\gamma x_2x_3+\gamma{\xi} x_2+\gamma(\gamma+1)x_3+\gamma(\gamma+1){\xi})\\
 & +(({\xi}+1)x_1x_2+(\gamma+1)x_1x_4+(\gamma+1)({\xi}+1)x_1\\
 &+\gamma x_2x_4+\gamma({\xi}+1)x_2+\gamma(\gamma+1)x_4+\gamma(\gamma+1)({\xi}+1))\\
&+\beta_{12}((\gamma+1)x_1+\gamma x_2+\gamma(\gamma+1))+\beta_{13}( {\xi} x_1+\gamma x_3+\gamma{\xi})\\
 &+(1+\beta_{13})(({\xi}+1)x_1+\gamma x_4+\gamma({\xi}+1))+\beta_{23}({\xi} x_2+(\gamma+1)x_3+(\gamma+1){\xi})\\
 &+(1+\beta_{23})(({\xi}+1)x_2+(\gamma+1)x_4+(\gamma+1)({\xi}+1))\\
 &+(({\xi}+1)x_3+{\xi} x_4+{\xi}({\xi}+1))\\
 = &
\left((\gamma+1)x_1+\gamma x_2+\gamma\beta_{13}+(\gamma+1)\beta_{23}+{\xi}+1)\right)\cdot \\
& \cdot  \left(x_3+x_4+\gamma x_1+(\gamma+1)x_2+{\xi}+\beta_{12}+(\gamma+1)\beta_{13}+\gamma\beta_{23}\right) + (\textit{constants})
\end{align*}
For any values of $(\gamma,{\xi})$, $D_af$ is equivalent to $x_1  x_2$, which is not balanced.
Then there are 4 distinct vectors such that the derivatives are not balanced.\\
Now consider the vector $\bar a=(1100)$. This is different from the four vectors above and the derivative in $a$ is 
\begin{align*}
D_{\bar a}f(x) = & x_1x_3+x_2x_3+x_3+x_1x_4+x_2x_4+x_4+\beta_{12}(x_1+x_2+1)\\
& \ +\beta_{13}x_3+(1+\beta_{13})x_4+\beta_{23}x_3+(1+\beta_{23})x_4\\
= & (x_3+x_4+\beta_{12})(x_1+x_2+1+\beta_{13}+\beta_{23})+\gb_{12}(\beta_{13}+\beta_{23}).
\end{align*}

The last expression is equivalent to $x_1x_2$, which is not balanced. So we have found five non-balanced derivates and therefore $|\Gamma(f)|<11$.

\item[\bf iv)] {\bf $|\mathrm{supp}_3(f)|=1$}

In this last case we can assume
 $$f(x) = x_1x_2x_3+\sum_{i<j}\beta_{ij}x_ix_j.$$
From the derivative in $e_1$ we have
 \begin{align*}
D_{e_1}f(x) = & x_2x_3+\sum_{j\ne1}\beta_{1j}x_j\\
 = & (x_2+\beta_{13})(x_3+\beta_{12})+\beta_{14}x_4+\beta_{12}\beta_{13},
 \end{align*}
  so $\beta_{14}=1$.\\
  We can obtain similarly $\beta_{24}=1$ and $\beta_{34}=1$ from the derivatives $D_{e_2}f$ and $D_{e_3}f$ respectively.
As in the previous case we want to find more than four (not-null) elements that give unbalanced derivatives.\\
Let us consider $a\neq{0}$, $a=(a_1,a_2,a_3,a_4)$, with $a_1+a_2+a_3=0$ and $(a_1,a_2,a_3)\neq(0,0,0)$.
Since in $f$ the first three variables take the same role, we can assume without loss of generality that $a_3=0$, so $a_1=a_2=1$.
Now let us consider the derivative of $f$ with respect to $a=(1\ 1\ 0\ a_4)$.
\begin{align*}
D_af(x) = & x_1x_3+x_2x_3+x_3+\beta_{12}(x_1+x_2+1)\\
 &+x_3(\beta_{13}+\beta_{23}+a_4)+(x_4+a_4x_1)+(x_4+a_4x_2)=\\
= & (x_3+\beta_{12}+a_4)\cdot(x_1+x_2+1+\beta_{13}+\beta_{23}+a_4)\\
 &\ + (\beta_{12}+a_4)(1+\beta_{13}+\beta_{23}+a_4).
\end{align*}

The obtained function is equivalent to $x_1x_2$, so it is not balanced. Therefore there are at least 6 elements for which $D_af$ is not balanced and so $\Gamma(f)$<10.
\end{itemize}
The analysis of the previous cases has shown the following:
\begin{itemize}
\item if $\mathrm{supp}_3(f)=4$ then $\Gamma(f)  < 10$,
\item if $\mathrm{supp}_3(f) =3$ then $\Gamma(f) < 10$,
\item if $\mathrm{supp}_3(f)=2$ then $\Gamma(f)  < 11$,
\item if $\mathrm{supp}_3(f)=1$ then $\Gamma(f) <10$.
\end{itemize}
We can conclude that for any $f$ we certainly have $\Gamma <11$.
\end{proof}

\section{\INS{Consequences in low dimensions}}\label{sec:4}

In this section we discuss some consequences of our previous results.

~\\
\textbf{m=4}\\

There are two immediate non-existence results for the case $m=4$.
The first holds also for non-bijective functions.

\begin{theorem}\label{th:cubic}
Let $m=4$ and $F$ a pure cubic \INS{vectorial Boolean function}. Then $F$ is not APN. 
\end{theorem}
\begin{proof}
Suppose that $F$ is APN. Then from \INS{Theorem} \ref{prop:gamma} there exists a component $F_\gl$ such that $|\Gamma(F_\gl)|\ge 11$. But Theorem \ref{lm:gamma} shows $\Gamma(F_\gl)<11$, since $F_\gl$ is cubic. 
\end{proof}

As consequence we obtain our second result, which \INSm{is the non-existence of APN permutations for $m=4$.}

\begin{corollary}
There is no APN permutation for $m=4$.
\end{corollary}
\begin{proof}
Since $F$ is invertible, then $F$ is at most cubic. Since $F$ is an APN, then $n_1(F)=0$ (Theorem \ref{th:linapn}).
Since $F$ is an invertible APN, then $n_2(F)=0$ (Corollary \ref{cor:quad}). Therefore, $F$ is a pure cubic.
But this contradicts Theorem \ref{th:cubic}.
\end{proof}

%

Observe that Theorem \ref{th:cubic} and our previous result do no prevent the existence of quadratic or cubic APN's, which are known to exist (but of course they are not bijective). Indeed, in \cite{CGC2-cry-art-brinkmann2008classification} they show (computationally) that there are exactly two APN functions in dimension 4 (up to EA-equivalence), one quadratic and one cubic. Theorem \ref{th:cubic} shows that the latter, being a cubic APN, cannot be a pure cubic.

~\\
\textbf{m=6}\\

The first consequence of our previous result to dimension $6$ is the following corollary.
\begin{corollary}
Let $F$ be an APN permutations in dimension $6$, the the degrees of its non-zero components
can be either $3$,$4$ or $5$.
\end{corollary}
\begin{proof}
An invertible APN cannot have any quadratic component (Corollary \ref{cor:quad}), any linear component (Theorem \ref{th:linapn})
and any degree-$6$ component (because it's invertible).
\end{proof}
This is in accordance with experimental results by Langevin \cite{CGC2-cry-art-lan} and explains also because
the degree of the components of Dillon's permutation are only $3$ and $4$ (although in principle
there could be a component of degree $5$). 

Theorem \ref{th:cubic} cannot be trivially extended to dimension $6$, since pure cubic APN functions
exist. 
Still, \INS{Theorem} \ref{prop:gamma} applies and therefore any cubic APN must have a component
with a large number of balanced derivatives, that is, with $\Gamma(F_\lambda) \geq 47$.
However, for $m=6$ this arises no contradiction, since there are cubics with $\Gamma(f) \geq 47$ and
a suitable extension of Theorem \ref{lm:gamma} to the case $m=6$ is missing.\\

We conclude this paper with noting that from the exam of the computational results presented in \cite{CGC2-cry-art-lan}, 
we can derive a weaker version of Theorem \ref{th:cubic} for $m=6$. 
\begin{theorem}
Let $m=6$ and $F$ an invertible pure cubic. Then $F$ is not APN.
\end{theorem}

~\\
\textbf{Acknowledgement}\\

We would like to thank C. Carlet for useful discussions, and L. Budaghyan for nice observations and for pointing out the result of Corollary \ref{cor:ea} for the case of quadratic components.
\bibliographystyle{plain}

\end{document}